\documentclass[conference,letterpaper]{IEEEtran}
\IEEEoverridecommandlockouts 
\usepackage[utf8]{inputenc} 
\usepackage[T1]{fontenc}
\usepackage{url}
\usepackage{cuted}
\usepackage{ifthen}
\usepackage{graphicx}
\usepackage{psfrag}
\usepackage{epsfig}
\usepackage{epstopdf}

\newcommand{\FIG}[1]{\begin{center}
  \mbox{\epsfclipoff\epsffile{#1.eps}}
  \end{center}}
\usepackage[cmex10]{amsmath} 
\DeclareMathOperator*{\argmin}{arg\,min}
\usepackage{tikz}
\usetikzlibrary{snakes}
\usetikzlibrary{shapes,arrows}
\usetikzlibrary{arrows,positioning,shapes.geometric}
\usepackage{amssymb}
\usepackage{amsthm}
\usepackage{bm}
\usepackage{bbm}
\usepackage{hyperref}
\hypersetup{ colorlinks=true, linkcolor=blue, citecolor=blue, urlcolor = blue}
\usepackage[linesnumbered,ruled,vlined]{algorithm2e}
\usepackage{cleveref}
\usepackage{comment}
\usepackage{mathtools}
\usepackage{caption}
\newtheorem{theorem}{Theorem}
\newtheorem*{blemma}{Burnashev's Lemma}
\newtheorem{lemma}{Lemma}

\newtheorem{definition}{Definition}

\begin{document}
\title{Instantaneous SED coding over a DMC} 


\author{ Nian Guo and Victoria Kostina
\thanks{N. Guo and V. Kostina are with the Department
of Electrical Engineering, California Institute of Technology, Pasadena, CA, 91125 USA. E-mail: \{nguo,vkostina\}@caltech.edu. This work was supported in part by the National Science Foundation (NSF) under grants CCF-1751356 and CCF-1956386.}}


\maketitle


\begin{abstract}
In this paper, we propose a novel code for transmitting a sequence of $n$ message bits in real time over a discrete-memoryless channel (DMC) with noiseless feedback, where the message bits stream into the encoder one by one at random time instants. Similar to existing posterior matching schemes with block encoding, the encoder in our work takes advantage of the channel feedback to form channel inputs that contain the information the decoder does not yet have, and that are distributed close to the capacity-achieving input distribution, but dissimilar to the existing posterior matching schemes, the encoder performs \emph{instantaneous} encoding - it immediately weaves the new message bits into a continuing transmission. A posterior matching scheme by Naghshvar et al. partitions the source messages into groups so that the group posteriors have a \emph{small-enough difference} (SED) to the capacity-achieving distribution, and transmits the group index that contains the actual message. Our code adopts the SED rule to apply to the evolving message alphabet that contains all the possible variable-length strings that the source could have emitted up to that time. Our instantaneous SED code achieves better delay-reliability tradeoffs than existing feedback codes over $2$-input DMCs: we establish this dominance both by simulations and via an analysis comparing the performance of the instantaneous SED code to Burnashev's reliability function.

Due to the message alphabet that grows exponentially with time $t$, the complexity of the instantaneous SED code is double-exponential in $t$. To overcome this complexity barrier to practical implementation, we design a low-complexity code for binary symmetric channels that we name the instantaneous \emph{type set} SED code. It groups the message strings into sets we call \emph{type sets} and tracks their prior and posterior probabilities jointly, resulting in the reduction of complexity from double-exponential to $O(t^4)$. Simulation results show that the gap in performance between the instantaneous SED code and the instantaneous type-set SED code is negligible.

\end{abstract}

\section{Introduction}

With the emergence of the Internet of Things (IoT), communication systems, such as those employed in distributed control and tracking scenarios, are becoming increasingly dynamic, interactive, and delay-sensitive. In classical feedback codes with block encoding \cite{Horstein}--\cite{Antonini}, the encoder knows the entire message sequence ahead of the transmission. Since collecting the message bits into a block prior to the transmission induces delay, these codes are ill-suited to real-time applications where the data continuously streams into the encoder in real time. The encoder in our work performs \emph{instantaneous} encoding: it starts transmitting as soon as the first message bit arrives and incorporates new message bits into the continuing transmission on the fly. 

Although feedback does not increase the capacity of a memoryless channel \cite{Shannon}, it simplifies the design of capacity-achieving codes and improves
achievable delay-reliability tradeoffs \cite{PPV}, \cite{Burnashev}. Horstein~\cite{Horstein} proposed a scheme for transmitting an infinite bit sequence over a binary-symmetric channel (BSC) with feedback. Horstein's scheme was shown to be capacity-achieving in \cite{Sch1}--\cite{Sch2}, \cite{Shayevitz}.
Waeber et al.~\cite{Waeber}, who refer to Horstein's scheme as `probabilistic bisection algorithm' (PBA), showed that its expected $L^1$-loss between the message and the estimate converges at a geometric rate with time. Schalkwijk and Kailath~\cite{SK-I} designed a feedback code that achieves the capacity of an additive white Gaussian noise channel under an average power constraint. Several feedback codes with block encoding for DMCs \cite{Burnashev}--\cite{Naghshvar2} have been proposed to attain Burnashev's reliability function \cite{Burnashev}, which is the maximum rate of the exponential decay of the error probability when the rate is strictly below the channel capacity and the blocklength is taken to infinity. Naghshvar et al.~ \cite{Naghshvar}--\cite{Naghshvar2} designed a feedback code for a binary-input DMC with feedback that is similar in spirit to Horstein's scheme \cite{Horstein}. Provided that the capacity-achieving distribution is Bernoulli$\left(\frac{1}{2}\right)$, Naghshvar’s scheme provably attains Burnashev’s reliability function \cite{Burnashev}. Naghshvar et al.'s `small-enough difference' (SED) encoder~\cite{Naghshvar} partitions the message alphabet into groups so that the groups' posterior probabilities are as close as possible to the capacity-achieving input distribution, and transmits the index of the group that contains the message. While the SED rule in~ \cite{Naghshvar} has double-exponential complexity in the message length $n$, Antonini et al.~\cite{Antonini} proposed a scheme for the BSCs with complexity $O(n^2)$. The complexity reduction is realized by grouping the elements with the same posterior, i.e., with the same Hamming distance to the received sequence.  Shayevitz and Feder~\cite{Shayevitz} abstracted the similarities between Horstein's scheme \cite{Horstein} and Schalkwijk-Kailath's scheme \cite{SK-I} into a general method to design capacity-achieving feedback schemes they termed \emph{posterior matching}. The posterior matching method ensures that the encoder transmits the information about the message that has not yet been seen by the decoder, and matches the distribution at the output of the encoder to the capacity-achieving distribution. It applies to most memoryless channels.

While the schemes in \cite{Horstein}--\cite{Shayevitz} use block encoding, two existing works \cite{Lalitha}, \cite{Antonini2} considered instantaneous encoding. Lalitha et al. \cite{Lalitha} assumed that the inter-arrival times of bits are known by the decoder and that the channel is a BSC.\footnote{These assumptions are significantly more restrictive than our setting, as our decoder only knows the distribution of the inter-arrival times, and our channel is a DMC.} Their scheme operates as follows. An infinite binary message sequence is mapped to a real number on $[0,1]$. The message prior is uniform before the first transmission. The message posterior is updated both by the encoder and the decoder upon receiving each channel output. At the $n$-th bit arrival time, the encoder and the decoder divide each of the existing $2^{n-1}$ uniform intervals in $[0,1]$ into two equal intervals. The encoder finds the interval that contains the median of the message posterior, and transmits a bit describing whether the real-valued representation of $n$ received bits is above or below a randomly chosen boundary of this interval. Lalitha et al.~\cite{Lalitha} proved that their scheme achieves the probability of erroneously decoding the first $n$ bits at time $t$ that decays exponentially in $t$, and derived a lower bound on the maximum rate that leads to a vanishing error probability. Antonini et al. \cite{Antonini2} designed a causal encoding scheme for streaming bits with a fixed arrival rate over a BSC and showed by simulation that the code rate approaches the channel capacity as the bit arrival rate approaches the transmission rate.

In this paper, we propose a feedback code with instantaneous encoding we term the \emph{instantaneous SED code}. At each time $t$, the encoder and the decoder first calculate the priors of all possible message strings using the message bit arrival probability and the posteriors at time $t-1$. Then, they apply an adapted SED rule to partition the evolving message alphabet into groups. The adapted SED rule drives the group priors instead of group posteriors as close as possible to the capacity-achieving distribution. This is necessary because due to the possible arrival of a new message bit at time $t$, the posteriors at time $t-1$ are insufficient to describe all possible variable-length message strings at time $t$. In contrast, feedback codes with block encoding \cite{Horstein}--\cite{Antonini} only need to consider the posteriors, since the block encoding implies that the priors at time $t$ are equal to the posteriors at time $t-1$.
The encoder transmits the index of the group that contains the true variable-length string it received so far. The posterior of each variable-length message string is calculated using the prior of that message string and the received channel output.

The instantaneous SED code outperforms the existing feedback coding schemes over $2$-input DMCs. By analyzing the performance of the instantaneous SED code, we derive a lower (achievability) bound on the reliability function of a $2$-input DMC with feedback over the class of feedback codes with instantaneous encoding. Our bound outperforms the achievability bounds derived from Burnashev's \cite{Burnashev} and Naghshvar et al.'s \cite{Naghshvar} schemes with block encoding. Furthermore, simulations show that the instantaneous SED code achieves a significantly larger rate than Naghshvar et al.'s \cite{Naghshvar} and Lalitha et al.'s \cite{Lalitha} schemes.

Since the cardinality of the message alphabet increases exponentially with time $t$ as new message bits arrive, and since the SED rule has exponential complexity in the size of the message alphabet, the complexity of the instantaneous SED code is double-exponential in time. We design a polynomial-complexity code we term \emph{the instantaneous type-set SED code}, which applies to the BSCs. The complexity of the instantaneous type-set SED code is $O(t^4)$, while the rate gap to the instantaneous SED code is negligible. The complexity reduction is achieved due to the following effects. First, the instantaneous type-set SED code classifies the elements in the message alphabet into type sets whose cardinality is $O(t^2)$, which is much smaller than $O(2^t)$, the cardinality of the alphabet at time $t$. Thus, performing the prior and the posterior updates in terms of the type sets instead of the individual strings leads to an exponential reduction in complexity. 
Second, the instantaneous type-set SED code uses a rule we refer to as the type-set SED rule where the alphabet is partitioned by grouping the type sets rather than the individual strings. The type-set SED rule has quadratic complexity in the number of type sets, which is exponentially smaller than the complexity of the SED rule. Although Antonini et al.'s scheme \cite{Antonini} attains complexity reduction also by grouping binary strings, the type sets used in the instantaneous type-set SED code are different from the groups used in \cite{Antonini}. While the groups in Antonini et al.'s scheme \cite{Antonini} are generated all at once at time $t=n$ by grouping the message strings that have the same Hamming distance to the received channel outputs, the instantaneous type-set SED code generates new type sets at every time $t\leq n$.

The rest of the paper is organized as follows. In Section~\ref{SecII}, we formally define feedback codes with instantaneous encoding. In Section~\ref{SecIII}, we present the instantaneous SED code, analyze its performance to derive an achievability bound to the reliability function, and show that it outperforms existing feedback codes. In Section~\ref{LRTSC_sec}, we present the instantaneous type-set SED code.

\textit{Notation:} We denote by $\{0,1\}^t$ the set of all binary strings of length equal to $t$, and we denote by $\mathcal B_t$ the set of all binary strings of length less than or equal to $t$,
\begin{align}
    \mathcal B_t\triangleq \cup_{1\leq s\leq t}\{0,1\}^s.
\end{align}
We denote by $\boxplus$ the concatenation operation between two strings, e.g. $101\boxplus 1=1011$. We denote by $\boxminus$ the truncation operation that deletes the last bit of a string, e.g. $101\boxminus = 10$. We define the function $[\cdot]_n\colon \mathcal B_{\infty}\rightarrow \mathcal \{0,1\}^n$
that extracts the first $n$ bits of its argument, e.g. $[1011]_2=10$. For a possibly infinite sequence $x=\{x_1,x_2,\dots\}$, we write $x^i=\{x_1,x_2,\dots,x_i\}$ to denote the vector of its first $i$ elements.

\section{Problem Formulation}\label{SecII}
Consider the real-time communication system in Fig.~\ref{problem}.
\begin{figure}[h!]
\centering
\includegraphics[trim = 25mm 240mm 62mm 32mm, clip, width=8cm]{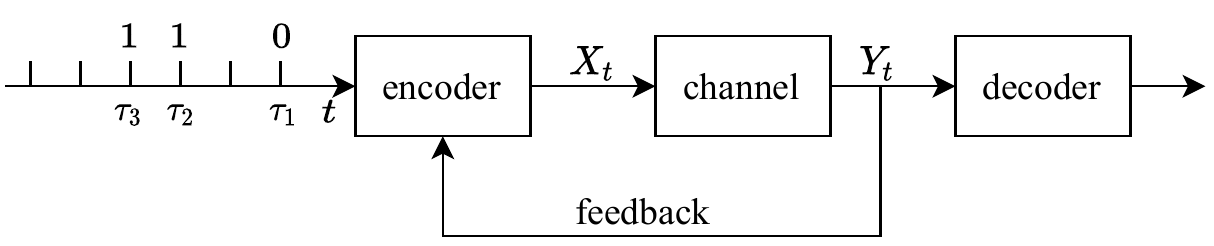}
\caption{Real-time communication over a channel with feedback}
\label{problem}
\end{figure}
A sequence of message bits streams into the encoder one by one at random time instants. At each time instant, the encoder may only receive a bit $0$ or a bit $1$ or no bits. We denote by $B_t^*$ the bit sequence that has arrived at the encoder by time $t$.  The length of $B_t^*$ is less than or equal to $t$: 
\begin{align}
    &B_t^*\in\mathcal B_t,~t=1,2,\dots
\end{align}
A new bit arrives at time $t+1$ according to the probability distribution
\begin{equation}\label{bit_prob}
P_{B_{t+1}^*|B_t^*}.
\end{equation}
Since at most one bit can arrive at any time, the conditional distribution $P_{B_{t+1}^*|B_t^*}(\cdot|s)$ can only place non-zero masses at $B_{t+1}^* = s$, $B_{t+1}^* = s\boxplus 0$ and $B_{t+1}^* = s\boxplus 1$.
Both the encoder and the decoder know \eqref{bit_prob}. We denote by $\tau_n$ the time at which the $n$-th message bit arrives at the encoder. $B_{\tau_n}^*$ is the length-$n$ message string that we aim to decode. Throughout, we assume that $\tau_1=1$.

The next definition formalizes causal feedback codes with instantaneous encoding.
\begin{definition}[An $(n,s,\epsilon)$ feedback code with instantaneous encoding]\label{def1}
Fix a bit arrival distribution \eqref{bit_prob} and a memoryless channel with a single-letter channel transition probability $P_{Y|X}\colon \mathcal X\rightarrow\mathcal Y$. An $(n,s,\epsilon)$ feedback code with instantaneous encoding consists of:\\
1. A sequence of encoding functions $\mathsf f_t\colon \mathcal B_t\times \mathcal Y^{t-1}\rightarrow \mathcal X$, $t=1,2,\dots$ that the encoder uses to form the channel inputs
\begin{align}\label{enc}
    X_t \triangleq \mathsf f_t(B_t^*, Y^{t-1});
\end{align}
2. A sequence of decoding functions $\mathsf g_{t}\colon \mathcal Y^t \rightarrow \{0,1\}^n$, $t=1,2,\dots$;\\
3. A stopping time $\lambda_n$ adapted to the filtration generated by $\{Y_t\}_{t=1,2,\dots}$ that determines when the transmission stops and that satisfies
\begin{align}\label{time_constraint}
    &\mathbb E [\lambda_n]\leq s,\\ \label{error_constraint}
    & \mathbb P[\hat B\neq B_{\tau_n}^*]\leq \epsilon,
\end{align}
where 
\begin{align}\label{dec}
    \hat B \triangleq \mathsf g_{\lambda_n}(Y^{\lambda_n}).
\end{align}

\end{definition}

For any $R>0$, the minimum error probability achievable by rate-$R$ feedback codes with instantaneous encoding and message length $n$ is given by
\begin{align}\nonumber
    \epsilon_n^*(R) \triangleq \min\{\epsilon\colon\exists~& \left(n,\frac{n}{R},\epsilon\right)~\text{feedback code}\\
    &\text{with instantaneous encoding}\}.
\end{align}
The reliability function $E(R)$ for feedback codes with instantaneous encoding is defined as
\begin{align}\label{reliabilityfunc}
    E(R) \triangleq \lim_{n\rightarrow\infty}\frac{R}{n}\log_2\frac{1}{\epsilon_n^*(R)}.
\end{align}

\section{Instantaneous SED Code}\label{SecIII}
In this section, we present the instantaneous SED code, which falls in the framework of Definition~\ref{def1}.
Then, we analyze the new code to provide an achievability bound on the reliability function \eqref{reliabilityfunc} of a $2$-input DMC with feedback over the class of codes in Definition~\ref{def1}. Finally, we numerically compare the instantaneous SED code with existing feedback coding schemes.
\subsection{The instantaneous SED code: algorithm}\label{RTC}
At time $t$, the encoder and the decoder first update the priors $P_{B_{t}^*|Y^{t-1}}$ of all elements in $\mathcal B_t$ using the posteriors $P_{B_{t-1}^*|Y^{t-1}}$ and the bit arrival distribution \eqref{bit_prob}. For any $y^{t-1}\in\mathcal Y^{t-1}$, the prior of the string $i\in \mathcal B_{t}$ at time $t$ is given by
\begin{align}\nonumber
    &P_{B_{t}^*|Y^{t-1}}(i|y^{t-1})\\\label{post_prior}
    =~&\sum_{j\in\mathcal B_{t-1}}P_{B_{t}^*|B_{t-1}^*}(i|j) P_{B_{t-1}^*|Y^{t-1}}(j|y^{t-1}),
\end{align}
where \eqref{post_prior} holds since $B_{t}^*-B_{t-1}^*-Y^{t-1}$ is a Markov chain.

Then, the encoder and the decoder partition the alphabet $\mathcal B_t$ into $|\mathcal X|$ disjoint groups $\{\mathcal G_x(t)\}_{x\in\mathcal X}$ using the prior $P_{B_t^*|Y^{t-1}}$, such that if 
\begin{align}\label{pXinequal}
  P_X^*(a)\geq   P_X^*(b),
\end{align}
then\footnote{If there exists $a,b\in \mathcal X$ such that \eqref{pXinequal} is satisfied with equality, then the encoder and the decoder can avoid ambiguity by agreeing on an order relation between $a$ and $b$ before the transmission begins.}
\begin{align}\label{Gxinequal}
   P_{B_t^*|Y^{t-1}}(\mathcal G_{a}(t)|y^{t-1})\geq  P_{B_t^*|Y^{t-1}}(\mathcal G_{b}(t)|y^{t-1}),
\end{align}
and for all other possible partitions $\{\mathcal G_{x}'(t)\}_{x\in\mathcal X}$
\begin{align}\nonumber
    &\sum_{x\in\mathcal X}\left|P_{B_t^*|Y^{t-1}}(\mathcal G_{x}(t)|y^{t-1}) - P^*_X(x)\right| \\\label{setparition}
    \leq& \sum_{x\in\mathcal X}\left| P_{B_t^*|Y^{t-1}}(\mathcal G_{x}'(t)|y^{t-1}) - P^*_X(x)\right|,
\end{align}
where $P_X^*$ is the capacity-achieving input distribution.

The output of the encoder is formed as follows. Once the message alphabet is partitioned, the encoder determines the group that contains the binary string $B_t^*$ it received so far and transmits the index of that group,
\begin{subequations}\label{channel_XT}
\begin{align}
    X_t &= \mathsf f_t(B_t^*,Y^{t-1})\\ \label{channel_XT_b}
    &\triangleq \sum_{x\in\mathcal X} x\mathbbm{1}\{B_t^*\in \mathcal G_x(t)\}.
\end{align}
\end{subequations}
The set partitioning rule in \eqref{pXinequal}--\eqref{setparition} together with the encoding function in \eqref{channel_XT} drives the distribution of $X_t$ as close as possible towards the capacity-achieving distribution.

 Upon receiving the channel output $Y_t$ at time $t$, the decoder updates the posteriors $P_{B_t^*|Y^t}$ using the priors $P_{B_t^*|Y^{t-1}}$ and the channel output $Y_t$. For any $y^t\in\mathcal Y^t$, the posterior of the string $i\in \mathcal B_t$ at time $t$ is given by
\begin{align}\label{prior_post}
    P_{B_t^*|Y^t}(i|y^t) =\frac{P_{Y|X}(y_t|\mathsf f_t(i,y^{t-1})) P_{B_t^*|Y^{t-1}}(i|y^{t-1})}{\sum_{x\in\mathcal X}P_{Y|X}(y_t|x)P_{B_t^*|Y^{t-1}}(\mathcal G_x(t)|y^{t-1})},
\end{align}
where \eqref{prior_post} holds since $Y_t-X_t-(B_t^*,Y^{t-1})$ is a Markov chain. Since the encoder knows $Y_t$ through the noiseless feedback, the posterior in \eqref{prior_post} is also available at the encoder.

 The decoding is performed at stopping time
\begin{align}\label{tau}
    \lambda_n \triangleq \min\left\{t\colon \max_{i\in\{0,1\}^n} P_{\left[B_t^*\right]_n|Y^t}(i|Y^t)\geq 1-\epsilon\right\}, \epsilon\in(0,1).
\end{align}
The decoder forms the estimate $\hat B$ at time $\lambda_n$ by choosing the binary string $i^*$ that achieves the maximum in the right side of \eqref{tau}. The stopping rule \eqref{tau} ensures that the error probability in \eqref{error_constraint} does not exceed $\epsilon$, see Appendix~\ref{appen_error} for details.

The set partitioning rule in \eqref{pXinequal}--\eqref{setparition} reduces to the SED rule \cite{Naghshvar} if the entire block of $n$ message bits is known at the encoder before the first transmission, i.e., 
\begin{align}
    B_t^* = B_{\infty}^*,~\forall t=1,2,\dots
\end{align}

\subsection{An achievability bound on the reliability function}\label{Ach}
We present a lower bound on the reliability function $E(R)$ in \eqref{reliabilityfunc} of a $2$-input DMC with feedback over the class of codes in Definition~\ref{def1}. 
We denote $\bar \tau_n \triangleq \mathbb E[\tau_n]$.
\begin{theorem}\label{prop1}
Fix a bit arrival distribution \eqref{bit_prob} that satisfies
\begin{itemize}
    \item[1)] $P_{B_{t+1}^*|B_t^*}(s|s) = 1,~\forall s\in\{0,1\}^n$;
    \item[2)] $\exists$ function $\mathsf d(\cdot)\colon\mathbb N\rightarrow \mathbb N$, such that  $d(n) = o(n)$ and
    \begin{align}\label{EN}
        \mathbb P[\tau_n\leq \bar \tau_n  + \mathsf d(n)] = 1
    \end{align}
\end{itemize}
and fix a $2$-input DMC with the single-letter transition probability $P_{Y|X}:\{0,1\}\rightarrow\mathcal Y$, $|\mathcal Y|<\infty$ such that
\begin{itemize}
    \item[3)] the channel capacity $C$ is achieved by
\begin{align}\label{cap_achi_2}
    P_{X}^*(0) = P_{X}^*(1) = 0.5;
\end{align}
\item[4)] the divergence
\begin{align}
    C_1\triangleq \max_{x_1,x_2\in\{0,1\}} D(P_{Y|X=x_1}|| P_{Y|X=x_2})
\end{align}
is maximized at $x_1 = 0, x_2 = 1$;
\item[5)] for any $y\in\mathcal Y$ and any $x\in\mathcal X$, the channel transition probability satisfies $P_{Y|X}(y|x)>0$.
\end{itemize} Then, the reliability function \eqref{reliabilityfunc} is lower bounded by
\begin{align}\label{lb_reliability}
    &E(R) \geq\\ \nonumber
    &C_1\left(1 - \lim_{n\rightarrow\infty}\left(\frac{H\left(B_{\bar \tau_n + \mathsf d(n)}^*|Y^{\bar \tau_n+\mathsf d(n)}\right)}{nC}+\frac{\bar \tau_n}{n}\right)R\right),
\end{align}
where $Y_1,Y_2,\dots$ are the channel outputs in response to the channel inputs generated by the encoder of the instantaneous SED code in Section~\ref{RTC}.
\end{theorem}
\begin{proof}
Appendix~\ref{Prof_thm1}.
\end{proof}

Assumption 1) posits that the entire message contains $n$ bits. Assumption 2) posits that the $n$-th bit arrival time $\tau_n$ is almost surely bounded by $\bar \tau_n + \mathsf d(n)$. Assumption 5) on the channel is implicitly required in \cite[Eq. (3.11)]{Burnashev} and \cite[Appendix B]{Naghshvar}.  We continue to assume 1)--5) in the discussion that follows.

Using Theorem~\ref{prop1}, we can quantify the advantage of the instantaneous SED code over the existing feedback codes with block encoding over a $2$-input DMC. To apply feedback codes with block encoding to the bit streaming setting, we let the encoder buffer the arriving bits until they form a block of $n$ bits and then transmit. The buffering induces delay $\tau_n$.
Using Burnashev's bound on $\mathbb E[\lambda_n-\tau_n]$  \cite[Theorems 1 and 2]{Burnashev}, where $\lambda_n$ is the decoding time, we immediately conclude that the buffer-then-transmit construction achieves the following lower bound
\begin{align}\label{NC_reliability}
    E(R) \geq C_1\left(1 - \left(\frac{1}{C}+\lim_{n\rightarrow\infty}\frac{\bar \tau_n}{n}\right)R\right).
\end{align}
Comparing the right sides of \eqref{lb_reliability} and \eqref{NC_reliability}, we notice that 1) the maximum achievable rates in \eqref{lb_reliability} and \eqref{NC_reliability} depend on both the channel transition probability and the bit arrival distribution; 2) since $B_{\bar \tau_n + \mathsf d(n)}^*$ and $Y^{\bar \tau_n + \mathsf d(n)}$ are dependent, $H(B_{\bar \tau_n + \mathsf d(n)}^*|Y^{\bar \tau_n + \mathsf d(n)})<n$, hence the lower bound in \eqref{lb_reliability} is at least as good as \eqref{NC_reliability}.  While it is not easy to obtain an analytical expression for $\lim_{n\rightarrow\infty}\frac{H(B_{\bar \tau_n + \mathsf d(n)}^*|Y^{\bar \tau_n + \mathsf d(n)})}{n}$ in general, one can calculate it numerically. For example, we calculate it over a BSC($0.02$) for periodic bit arrivals, i.e., $ P_{B_{t+1}^*|B_{t}^*}(s\boxplus b|s) = 0.5, b\in\{0,1\},t<n$, and we obtain
\begin{subequations}\label{HB}
\begin{align}\label{HB1}
    &\lim_{n\rightarrow\infty}\frac{H\left(B_{\bar \tau_n+ \mathsf d(n)}^*|Y^{\bar \tau_n + \mathsf d(n)}\right)}{n}\\ 
    = ~&\lim_{n\rightarrow\infty}\frac{H(B_{n}^*|Y^{n})}{n} = 0.145,
\end{align}
\end{subequations}
thus the gap between the right sides of \eqref{lb_reliability} and \eqref{NC_reliability} is significant.

To further quantify the delay-reliability tradeoffs attained by the instantaneous SED code, in Fig.~\ref{rate_HB_fig}~and Fig.~\ref{rate_fig2}, we plot the rate $R_n(\epsilon) = \frac{n}{\mathbb E[\lambda_n]}$ achieved by the instantaneous SED code in comparison with the existing schemes as a function of message length $n$, for a fixed error probability $\epsilon$. The instantaneous SED code achieves a significantly larger rate than either Naghshvar et al.' block encoding scheme \cite{Naghshvar} or Lalitha et al.'s instantaneous encoding scheme \cite{Lalitha}.
\begin{figure}[htbp]
\includegraphics[scale=0.4]{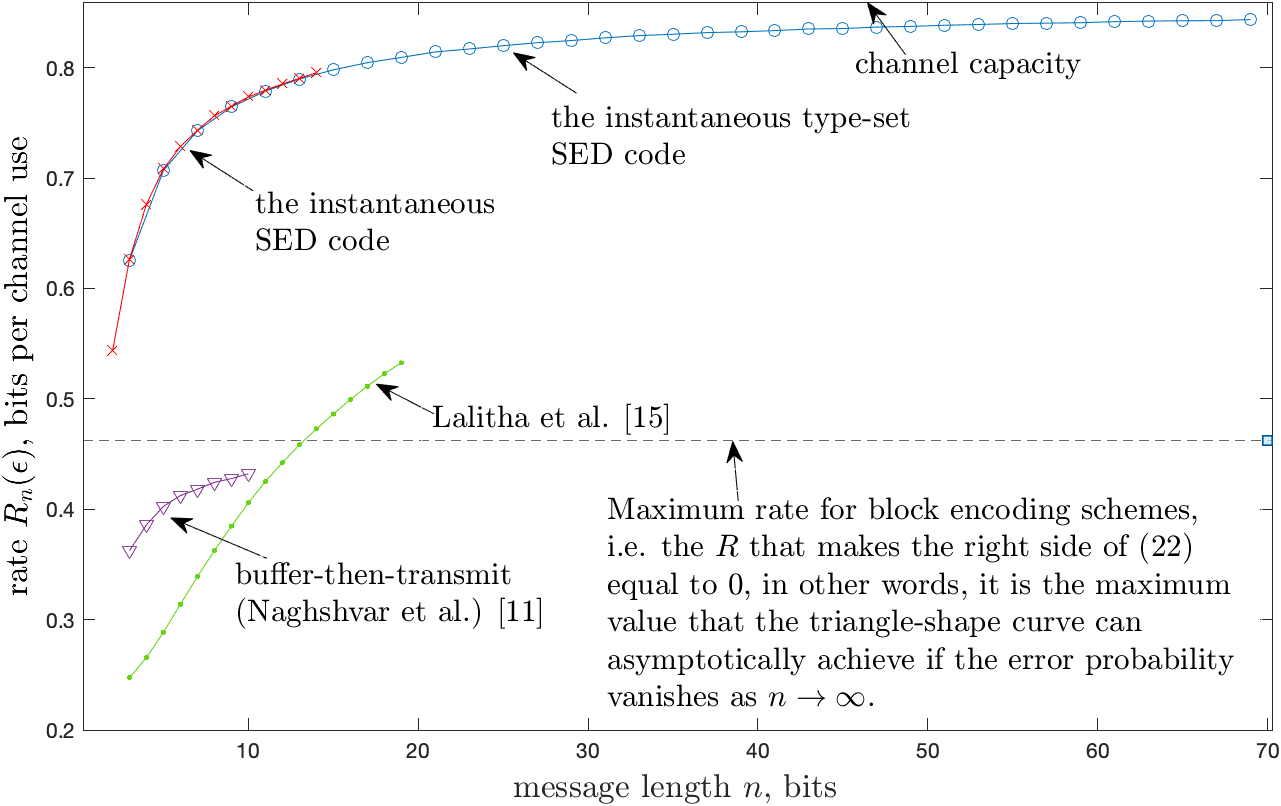}
\centering
\caption{Rate vs. message length over a BSC($0.02$). The error probability is constrained by $\epsilon = 10^{-3}$ \eqref{error_constraint}. At each time $t=1,2,\dots,n$, a new Bernoulli$\left(\frac{1}{2}\right)$ bit arrives at the encoder. We choose this bit arrival distribution to enable a comparison with Lalitha et al.'s code \cite{Lalitha}, which requires the knowledge of the bit arrival times at the decoder.
The curves are displayed for the range of $n$'s where the complexity of their calculation is not prohibitive.}
\label{rate_HB_fig}
\end{figure}
\begin{figure}
    \centering
    \includegraphics[scale=0.39]{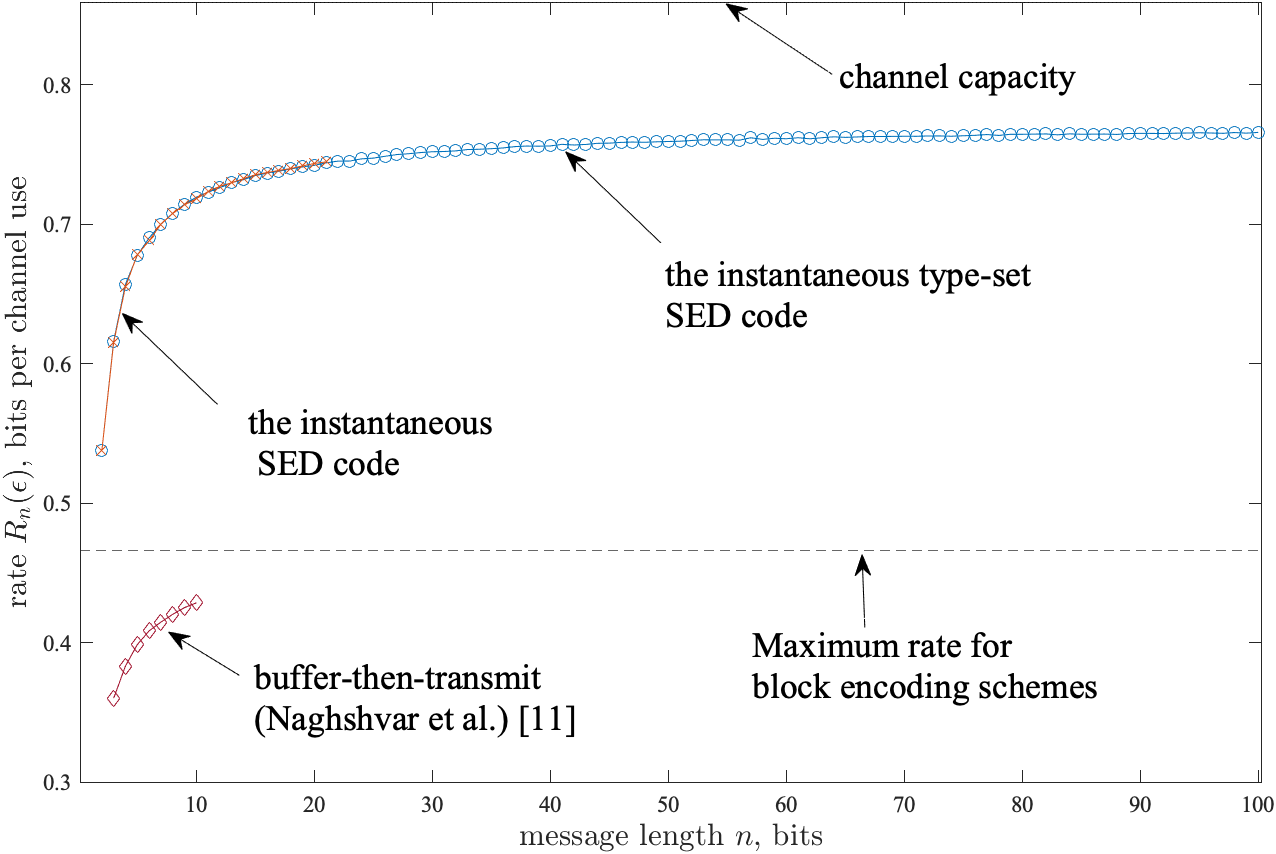}
    \caption{Rate vs. message length over a BSC($0.02$). The error probability is constrained by $\epsilon = 10^{-3}$ \eqref{error_constraint}. Message bits arrive at the encoder following a Bernoulli process in \eqref{equiprob} with $q=0.98$.
The curves are displayed for the range of $n$'s where their complexity is not prohibitive.}
    \label{rate_fig2}
\end{figure}
Note that the simulation results in Fig.~\ref{rate_fig2} show that the instantaneous (type-set) SED code continues to outperform Naghshvar et al.'s \cite{Naghshvar} scheme even if the message bits arrive at the encoder following a Bernoulli process, which does not satisfy assumption 2).
\section{Instantaneous type-set SED code}\label{LRTSC_sec}
In this section, we describe the instantaneous type-set SED code, whose complexity is polynomial in time. We assume that $n$ independent equiprobable message bits arrive at the encoder following a Bernoulli-$q$ process, i.e., $\forall b\in\{0,1\}$,
\begin{subequations}\label{equiprob}
\begin{align}
&P_{B_1^*}(0) = P_{B_1^*}(1) =0.5,\\
&P_{B_{t+1}^*|B_t^*} (s\boxplus b|s) =
    \begin{cases}
    \frac{q}{2}, &\text{if}~ s\in\mathcal B_{n-1},\\
    0, &\text{if}~ s\in\{0,1\}^n;
    \end{cases}
\end{align}
\end{subequations}
and assume that the channel is a BSC\footnote{The bit arrival distribution must satisfy that $P_{B_{t+1}^*|B_t^*} (s\boxplus 0|s) = P_{B_{t+1}^*|B_t^*} (s\boxplus 1|s)$, otherwise, the binary strings in a type set will not share the same prior probability. The instantaneous type-set SED code continues to apply if $q$ in \eqref{equiprob_b} is time-varying and/or has memory.}. 

Let $\mathcal S_1,\mathcal S_2,\dots$ be a sequence of sets, each containing some binary strings. We call these sets \emph{type sets}. We call a binary string $s_1$ the \emph{parent} of a binary string $s_2$ if $s_1 = s_2\boxminus$, and call a set $\mathcal S_k$ the \emph{parent} of a set $\mathcal S_j$ if all the parents of the strings in $\mathcal S_j$ are in $\mathcal S_k$. We denote by $p(j)$ the index of the parent of $\mathcal S_j$, e.g. $p(j) = k$. 

The following two insights are key to how the instantaneous type-set SED code attains complexity reduction:
\begin{itemize}
    \item[(a)]Binary strings in the alphabet are classified into type sets such that, at each time instant, the binary strings within a type set share the same length, same prior, and same posterior, and each type set has only one parent type set; thus type sets form a tree structure.  
    \item[(b)] A type-set SED rule is employed, which mimics the SED rule in \eqref{setparition}, except the groups $\{\mathcal G_x(t)\}_{x\in\{0,1\}}$ are formed by moving type sets rather than individual binary strings.
\end{itemize}
 
We proceed to show how the type sets in (a) are created at each time instant and why the type sets induce the complexity reduction. We denote by $\gamma_k(t)$ and $\rho_k(t)$ the prior and the posterior probabilities of the binary strings in $\mathcal S_k$ at time $t$. 
At time $t=1$, before the group partitioning, the encoder and the decoder create two type sets $\mathcal S_1\triangleq\{0\}$, $\mathcal S_2 \triangleq \{1\}$.
\begin{itemize}
    \item[1.] \emph{Creating new type sets.} At time $2\leq t \leq n$, before the group partitioning, each type set created at time $t-1$ generates a new type set by appending a $0$ and a $1$ to every binary string in that type set. For example, $\mathcal S_1$ and $\mathcal S_2$ generate $\mathcal S_3 = \{00,01\}$ and $\mathcal S_4 = \{10,11\}$.
    \item[2.] \emph{Updating priors.} As a consequence of classifying the strings into type sets, the prior $\gamma_k(t),k=1,2,\dots, t=2,3,\dots$ \eqref{post_prior} of the binary strings in $\mathcal S_k$ is fully determined by $\rho_k(t-1)$, $\rho_{p(k)}(t-1)$, and the bit arrival probability \eqref{equiprob}.
\end{itemize}
As we will see below, the type-set SED rule in (b) ensures that the binary strings within any type set created at time $t-1$ share the same posterior. Due to \eqref{equiprob}, the priors of all the strings in a type set can be updated simultaneously. Therefore, the binary strings in a type set at time $t$ share the same prior.

We proceed to introduce the type-set SED rule in (b) and to show why it leads to low complexity. 
\begin{itemize}
    \item[3.] \emph{Partitioning $\mathcal G_0(t)$ and $\mathcal G_1(t)$.} At time $t=1,2,\dots$, the encoder and the decoder sort all the existing type sets in descending order of $\gamma_k(t),k=1,2,\dots$. Starting from the top of the list, type sets are moved one by one to $\mathcal G_0(t)$ until the prior of $\mathcal G_0(t)$ exceeds $0.5$. The remaining type sets are moved to $\mathcal G_1(t)$. To drive the priors of the two groups closer to $(0.5,0.5)$, the encoder and the decoder may further partition the last type set moved to $\mathcal G_0(t)$ into two subsets with one left in $\mathcal G_0(t)$ and one moved to $\mathcal G_1(t)$. Due to the type-set creating process described in step 1, the binary strings in each type set can be sorted lexicographically into an ordered sequence. We split a type set by cutting the ordered sequence into two halves forming two subsets\footnote{To do the splitting, we do not need to sort the binary strings in type sets lexicographically at each time. What we do is we assign to every binary string in $\mathcal B_n$ a natural number in a lexicographic order. We can fully specify a type set by two numbers $a$ and $b$ corresponding to respectively the starting string $s_{\mathrm{start}}$ and the ending string $s_{\mathrm{end}}$ of that type set. Similarly, we can fully specify the new type set generated by this type set using $(2a+1,2b+2)$ to represent $(s_{\mathrm{start}}\boxplus 0,s_{\mathrm{end}}\boxplus1)$, and fully specify the two split subsets by $(a,a+n^*-1)$ and $(a+n^*,b)$.}. The number of strings in the subset moved to $\mathcal G_1(t)$ is $n^*$,
\begin{align}\label{nstar}
&n^* \triangleq \argmin_{n^*\in\{n_1,n_2\}}\left| 2 P_{B_{t}^*|Y^{t-1}}(\mathcal G_0(t)|y^{t-1}) - 1 -2n^*\gamma_k(t)\right|,\\
&n_1 \triangleq \left\lfloor\frac{P_{B_{t}^*|Y^{t-1}}(\mathcal G_0(t)|y^{t-1})-0.5}{\gamma_k(t)}\right\rfloor,\\
&n_2 \triangleq \left\lceil\frac{P_{B_{t}^*|Y^{t-1}}(\mathcal G_0(t)|y^{t-1})-0.5}{\gamma_k(t)}\right\rceil,
\end{align}
where $\gamma_k(t)$ is the prior of the binary string in the split type set.
If the type-set split results in the prior of $\mathcal G_0(t)$ becoming smaller than that of $\mathcal G_1(t)$, we simply swap $\mathcal G_0(t)$ and $\mathcal G_1(t)$ to satisfy \eqref{Gxinequal} with $a=0$, $b=1$. The type sets whose parent set was split may have two parents. To ensure that each type set has one valid parent, we recursively search for the type sets whose binary strings have parents from more than one type set and split them accordingly. This step guarantees that the posterior $\rho_{p(k)}(t)$, $k=1,2,\dots$, $t=1,2,\dots$ is deterministic.
\item[4.] \emph{Updating posteriors.} The posterior $\rho_k(t)$ $k=1,2,\dots$, $t=1,2,\dots$ \eqref{prior_post} only depends on $\gamma_k(t)$, the channel transition probability, and $Y_t= y_t$.
\item[5.] \emph{Stopping and decoding.} The stopping rule and the decoder of the instantaneous typeset
SED code operate as follows. The iterations stop at $t$ if
the type set whose binary string has the maximum posterior,
say $\rho_i(t)$, satisfies $\rho_i(t)\geq 1-\epsilon$, $|\mathcal S_i| = 1$, and the length of the only string in $\mathcal S_i$ is $n$. The decoder designates that string as the estimate $\hat B$.
\end{itemize}

A heuristic analysis in Appendix~\ref{complexity_algo} shows that the number of type sets at time $t$ is $O(t^2)$. The type-set SED rule sorts and splits type sets resulting in a quadratic complexity in the number of type sets, i.e., $O(t^4)$. The prior and the posterior updates together with the stopping and the decoding operations result in a linear complexity in the number of type sets, i.e., $O(t^2)$.
Therefore, the complexity of the instantaneous type-set SED code is $O(t^4)$. Fig.~\ref{rate_HB_fig} shows that the rate gap between the instantaneous SED code and the instantaneous type-set SED code is negligible. 

\section{Conclusion}
This is the first paper that implements instantaneous encoding of streaming data transmitted over a feedback channel with the bit arrival times unknown to the decoder. The new
algorithm -- the instantaneous SED code -- attains significantly better delay-reliability tradeoffs than existing schemes over $2$-input DMCs. This is evidenced both by the analysis of the reliability function (Theorem~\ref{prop1}) and by the simulation results (Fig.~\ref{rate_HB_fig}). While the instantaneous SED code has double-exponential compelxity, the new instantaneous type-set SED code has complexity of only $O(t^4)$ and no visible penalty in performance compared to the instantaneous SED code (Fig.~\ref{rate_HB_fig}). Analyzing the performance
of the instantaneous SED code over a general
DMC and providing a converse bound to the reliability function in \eqref{lb_reliability} are potential directions for future research. 

\section*{Acknowledgement}
Stimulating discussions with Prof. Richard Wesel, Amaael Antonini, and Prof. Aaron Wagner are gratefully acknowledged.

\appendix
\subsection{Error probability}\label{appen_error}
The error probability \eqref{error_constraint} of the instantaneous SED code is upper bounded by $\epsilon$ since
\begin{align}
    \{\hat B = \left[B_{\lambda_n}^*\right]_n\}~& = \{\hat B = \left[B_{\lambda_n}^*\right]_n, \left[B_{\lambda_n}^*\right]_n = B_{\tau_n}^*\}\\
    & \subseteq \{\hat B = B_{\tau_n}^*\}.\\
\hspace*{-0.25cm} \mathbb P\left[\hat B = \left[B_{\lambda_n}^*\right]_n\right]~
    & = \mathbb E\left[\max_{i\in\{0,1\}^n}\mathbb P\left[\left[B_{\lambda_n}^*\right]_n = i|Y^{\lambda_n}\right]\right]\\
    &\geq 1-\epsilon.
\end{align}
\subsection{Proof of Theorem~\ref{prop1}}\label{Prof_thm1}
For brevity, we define
\begin{align}
    &\rho_t(i,y^t) \triangleq P_{B_{t}^*|Y^t}\left(i|Y^t=y^t\right),\\
    &\theta_{t+1}(i,y^t) \triangleq P_{B_{t+1}^*|Y^t}\left(i|Y^t=y^t\right).
\end{align}

Together with the definition of $E(R)$ in \eqref{reliabilityfunc}, it suffices to show the following upper bound on the expected decoding time $\mathbb E[\lambda_n]$ to conclude \eqref{lb_reliability} in Theorem~\ref{prop1},
\begin{align} \nonumber
\mathbb E[\lambda_n] &\leq \frac{H(B_{\bar \tau_n +\mathsf d(n)}^*|Y^{\bar \tau_n +\mathsf d(n)})}{C}+\frac{\log_2 \frac{1}{\epsilon}}{C_1} \\\label{exp_tau}
&+ \bar\tau_n+ \mathsf d(n) + K,
\end{align}
where $K$ is a constant. For any $i\in\{0,1\}^n$, we define the random stopping time $\nu_i$ as
\begin{align} \label{nu}
  \nu_i\triangleq\min\left\{t\geq 0 \colon \rho_{t+\bar\tau_n+\mathsf d(n)}(i,Y^{t+\bar\tau_n+\mathsf d(n)})\geq 1-\epsilon\right\}.
\end{align}
According to the stopping rule in \eqref{tau}, for any $i\in\{0,1\}^n$,
\begin{align}
    \lambda_n \leq \nu_i + \bar \tau_n + \mathsf d(n),
\end{align}
which implies that the following inequality holds
\begin{align}
    &\mathbb E[\lambda_n]\\
    \leq ~& \bar \tau_n + \mathsf d(n)  + \mathbb E[\nu_i] \\ \label{v_tau}
    =~&  \bar \tau_n + \mathsf d(n) + 
   \mathbb E[\mathbb E[\nu_i|B_{\bar \tau_n+\mathsf d(n)}^*, Y^{\bar \tau_n+\mathsf d(n)}]].
\end{align}
To obtain the upper bound in \eqref{exp_tau}, it suffices to upper bound the conditional expectation in \eqref{v_tau}. Due to \eqref{EN} in assumption~(2), for $t\geq \tau_n + \mathsf d(n)$, we have for all $j\in\mathcal B_{n-1}$,
\begin{align}
    \theta_t(j,Y^{t-1}) = \rho_t(j,Y^t) = 0,~ a.s.
\end{align}
The instantaneous SED code in Section~\ref{SecIII} reduces to Naghshvar et al.'s scheme for $t\geq \tau_n + \mathsf d(n)$, that is, we consider $t=\tau_n + \mathsf d(n)$ as a new start in a classical block encoding setting, and we aim to transmit a message in $\{0,1\}^n$ over a feedback channel with the initial belief of each string $i\in \{0,1\}^n$ equal to $\rho_{\bar \tau_n + \mathsf d(n)}(i,Y^{\bar \tau_n + \mathsf d(n)})$. We use \cite[Eq. (14)]{Naghshvar} to upper bound the conditional expectation in \eqref{v_tau} as
\begin{align}\nonumber
    &\mathbb E[\nu_i(\bar \tau_n + \mathsf d(n))|B_{\bar \tau_n+\mathsf d(n)}^* = i, Y^{\bar \tau_n + \mathsf d(n)}=y]\\ \label{vi_ub}
    \leq~& \frac{\log\frac{1}{\epsilon}}{C_1} + \frac{\log\frac{1}{\rho_{\bar \tau_n + \mathsf d(n)}(i,y)}}{C} + K_i,
\end{align}
where $i\in\{0,1\}^t$, $y\in\mathcal Y^{\bar \tau_n+\mathsf d(n)}$, $K_i$ is a constant. Plugging \eqref{vi_ub} into \eqref{v_tau}, we obtain \eqref{exp_tau}.

\subsection{Number of type sets in the instantaneous type-set SED code}\label{complexity_algo}
We aim to show that the number of type sets at time $t$ is $O(t^2)$.

We define the following notations. We use $B$ or $A$ as the index of a random variable to signify that the random variable is obtained \emph{before} or \emph{after} the group partitioning process (step 3, Section~\ref{LRTSC_sec}). We denote by $N_B(t)$ and $N_A(t)$ the number of existing type sets at time $t$ before and after the encoder and the decoder partition $\mathcal G_0(t)$ and $\mathcal G_1(t)$, respectively.
We denote by $\mathcal S^*(t)$ the last type set moved to $\mathcal G_0(t)$ at time $t$ so that after $\mathcal S^*(t)$ is moved to $\mathcal G_0(t)$, the prior of $\mathcal G_0(t)$ exceeds $0.5$ for the first time (See step 3 in Section~\ref{LRTSC_sec}).
We denote by $\mathcal W_{B}(t)$ the set that contains all the new type sets created at time $t$ right before the encoder and the decoder partition $\mathcal G_0(t)$ and $\mathcal G_1(t)$ (See step 1 in Section~\ref{LRTSC_sec}, e.g., $\mathcal W_{B}(1) = \{\mathcal S_1,\mathcal S_2\}$, $\mathcal W_{B}(2) = \{\mathcal S_3,\mathcal S_4\}$). After the group partitioning process, some type set in $\mathcal W_{B}(t)$ may be split. Let $\mathcal W_{A}(t)$ be the set that consists of all subsets of the split type sets in $\mathcal W_{B}(t)$ and all unsplit type sets in $\mathcal W_{B}(t)$ after the group partitioning.

We show by heuristic analysis that the average number of type sets at time $t+1$ is upper bounded as
\begin{align}\label{Nth}
    &\mathbb E[N_B(t+1)] \leq \frac{2-q}{2}t^2 + \left(3-\frac{q}{2}\right)t + q, ~t\geq 1,\\ \label{Nth2}
    &\mathbb E[N_A(t+1)] \leq \mathbb E[N_B(t+1)] + (1-q)t +1,
\end{align}
where $q$ is the bit arrival probability in \eqref{equiprob}. 

We define a sequence of events $\mathcal E_t$, $t=1,2,3,\dots$ as
\begin{align} 
\mathcal E_1 \triangleq~& \{\mathcal S^*(1) \text{is not split}\}.\\ \nonumber
    \mathcal E_t \triangleq~ & \{\mathcal S^*(t) \text{ is split}\} \cap \{\text{the binary strings in }\\\label{E_event}
    &\mathcal S^*(t) \text{ are of length} \min(\lfloor qt\rfloor,n)\}, t>1.
\end{align}
Since at $t=1$, $|\mathcal S_1|=|\mathcal S_2| = 1$, $\mathbb P[\mathcal E_1] = 1$.
Extensive simulations on the evolution of type sets show that
\begin{align}
\mathbb P[\mathcal E^c_t] \ll 1, t=2,3,\dots
\end{align}
In the heuristic analysis that follows, we assume 
\begin{align}\label{heur_assump}
    \mathbb P[\mathcal E_t] = 1, t= 1,2,\dots
\end{align}

We will use rigorous analyses i)-iv) below together with the assumption in \eqref{heur_assump} to justify \eqref{Nth}--\eqref{Nth2}. The analyses ii)-v) below solely follow the construction of type-set SED codes. The type-set creating method in Section~\ref{LRTSC_sec} implies:
\begin{itemize}
    \item[i)] The binary strings in a type set are of the same length and can be ordered in a consecutive lexicographic order, e.g., $\mathcal S_3 = \{00,01\}$. To ensure that each type set has only one parent, once a type set is fixed to be split, among all its child type sets\footnote{We call $\mathcal S_j$ a child type set of $\mathcal S_i$ if $\mathcal S_i$ is the parent of $\mathcal S_j$. According to step 1 in Section~\ref{LRTSC_sec}, any type set $\mathcal S_i$ at most generates 1 type set $\mathcal S_j$. If $\mathcal S_j$ is never split, $\mathcal S_i$ has only one child type set $\mathcal S_j$. Yet, if $\mathcal S_j$ is split during the group partitioning process, $\mathcal S_i$ has multiple child type sets.}, at most one of them needs to be split accordingly. This is due to the reason that follows. Without loss of generality, we assume that we split an arbitrary type set $\mathcal S_k$ that contains $m$ binary strings $s_1,s_2,\dots,s_m$ sorted in a lexicographic order. 
    The type set splitting method (step 3, Section~\ref{LRTSC_sec}) cuts $\mathcal S_k$ between $s_{n^*}$ and $s_{n^*+1}$ to split it. Among all child type sets of $\mathcal S_k$, only the type set that contains both $s_{n^*}\boxplus 1$ and $s_{n^*+1}\boxplus 0$ will need to be split accordingly, and at most one type set $\mathcal S_m$ contains both two strings simultaneously. Recursively, due to the split of $\mathcal S_m$, the encoder and the decoder at most further split one child type set of $\mathcal S_m$. The recursion stops if the split type set has no child type sets, or if after the split of a type set $\mathcal S_k$, no child type sets of $\mathcal S_k$ contain $s_{n^*}\boxplus 1$ and $s_{n^*+1}\boxplus 0$ simultaneously. 
    
\item[ii)] For $t+1\leq n$, only the type sets in $\mathcal W_A(t)$ generate new type sets at time $t+1$ right before the group partitioning. These new type sets form $\mathcal W_{t+1,B}$. For $t+1>n$, $\mathcal W_{t+1,B}$ and $\mathcal W_{t+1,A}$ are empty since new type sets are no longer generated at each time $t+1$ before the group partitioning.
\item[iii)]By the definition of $\mathcal W_A(t)$, at time $t$, if a type set to be split is originally in $\mathcal W_{t,B}$, the two split subsets are in $\mathcal W_A(t)$. The binary strings of any type sets in $\mathcal W_B(t)$ are of length $\max(t,n)$.
\end{itemize}

Using \eqref{E_event} and i), we know:
\begin{itemize}
    \item[iv)] Given $\mathcal E_t$ occurs, after the split of $\mathcal S^*(t)$, the encoder and decoder at most further split $t-\lfloor qt\rfloor$ type sets. This is because (i) implies that given the split of $\mathcal S^*(t)$, the encoder and the decoder further split recursively at most $1$ type set of string length equal to $\lfloor qt \rfloor +1, \lfloor qt \rfloor +2, \dots, \max(t,n)$.
\end{itemize}
Using ii)--iv), we conclude that the encoder and the decoder at most split one type set in $\mathcal W_B(t)$, and the average cardinality of $\mathcal W_B(t), t=1,2,\dots$ evolves as
\begin{subequations}\label{cardw}
\begin{align}\label{card_wt}
    &\mathbb E[|\mathcal W_{B}(t+1)||\mathcal E_{t}] \leq \mathbb E[|\mathcal W_B(t)||\mathcal E_t] +1,~t\geq 2,\\ \label{card_wt2}
    &|\mathcal W_{B}(1)| = |\mathcal W_{B}(2)| = 2,
\end{align}
\end{subequations}
where \eqref{card_wt2} holds since no type set is split at $t=1$. Using ii) and iv), we conclude that given $\mathcal E_t$, the average number of type sets evolves as
\begin{subequations}
\begin{align}\nonumber
 \mathbb E[N_{B}(t+1)|\mathcal E_t]&\leq \mathbb E[N_B(t)|\mathcal E_t] + 1 + t-\lfloor qt\rfloor\\ \label{nt+1nt}
 &+ \mathbb E[|\mathcal W_{B}(t+1)||\mathcal E_t],\\
    N_B(1)&=2.
\end{align}
\end{subequations}
Replacing $\lfloor qt \rfloor$ by $qt$ in \eqref{nt+1nt}, plugging \eqref{cardw} into \eqref{nt+1nt}, and using \eqref{heur_assump}, we obtain \eqref{Nth}. Using \eqref{heur_assump}, iv), and \eqref{Nth}, we obtain \eqref{Nth2}.
\begin{figure}
    \centering
    \includegraphics[scale=0.5]{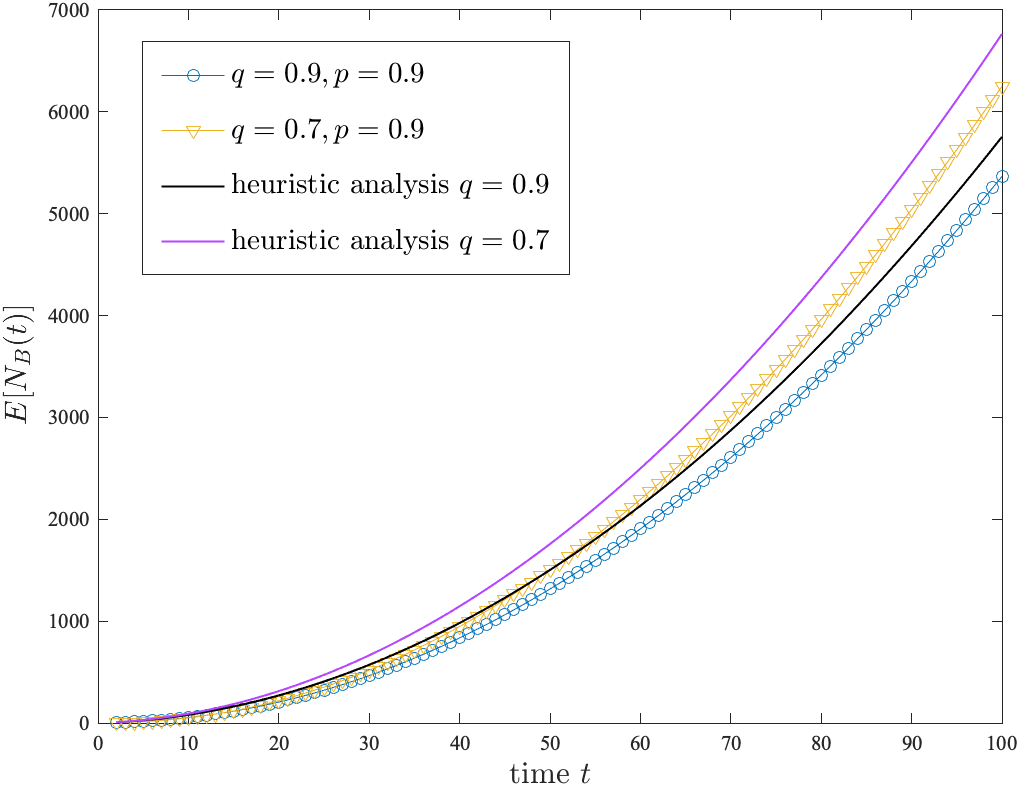}
    \caption{Average number of type sets $\mathbb E[N_B(t)]$ versus time $t$ over a BSC($p$). The bit arrival probability is $q$ in \eqref{equiprob}. The message length $n=100$. The curves by heuristic analysis are plotted as \eqref{Nth}. We only present curves for $p=0.9$, since according to our heuristic analysis, the upper bounds to the average number of type sets \eqref{Nth}--\eqref{Nth2} are not functions of $p$.}
    \label{groupt}
\end{figure}

Simulation results confirm our heuristic analysis. The fitting curves \eqref{Nth} in Fig.~\ref{groupt} increase at a similar speed as the simulated curves, indicating that the heuristic expressions in \eqref{Nth}--\eqref{Nth2} are meaningful gauges of the average number of type sets. The fitting curves in Fig.~\ref{groupt} are slightly larger than the simulated curves since \eqref{Nth}--\eqref{Nth2} are upper bounds to $\mathbb E[N_{B}(t+1)]$ and $\mathbb E[N_A(t+1)]$.

\end{document}